\documentclass[12pt]{amsart}
\usepackage{amssymb,latexsym,amsmath,amsthm,enumitem,mathrsfs,geometry}
\usepackage{fullpage}
\usepackage{fancyhdr}
\usepackage{xcolor}
\usepackage{bbm}
\usepackage{stmaryrd}
\usepackage{mathrsfs}
\usepackage{hyperref}
\usepackage{lineno}
\usepackage{diagbox}
\usepackage{array}
\usepackage{tikz}
\usetikzlibrary{arrows}
\usetikzlibrary{positioning}
\usepackage{float}
\usepackage{comment}
\usepackage{mathtools}
\usepackage{cleveref}
\usepackage{graphicx}
\usepackage{subcaption}

\newtheorem{theorem}{Theorem}[section]
\newtheorem*{theorem*}{Theorem}

\newtheorem*{remark*}{Remark}

\newtheorem{definition}{Definition}[section]

\begin{document}

\numberwithin{equation}{section}

\title{Lie symmetries of a generalized Fisher equation in cylindrical coordinates}

\author[Bayaraa]{Bayarjargal Batsukh}
\address{Department of Mathematics, National University of Mongolia, Ulaanbaatar, Mongolia}
\email{bayarjargalb@num.edu.mn}
\author[uuganaa]{Uuganbayar Zunderiya}
\address{Department of Mathematics, National University of Mongolia, Ulaanbaatar, Mongolia}
\email{zunderiya@gmail.com}

\thanks{}


\begin{abstract} In this work we studied a generalized Fisher equation in cylindrical coordinate using Lie symmetry method. We have determined for what type of source function the generalized Fisher equation has Lie Symmetries other than time translation symmetry when the diffusion function is given by an exponential function. Also the reduced ordinary differential equations are obtained corresponding to Lie symmetries of the generalized Fisher equation.
\end{abstract}

\maketitle

\section{Introduction}


In the last five decades, there appears a new trend using Lie point symmetry method to exploit the invariance property of differential equations. One can find more information on this method in \cite{olver,bluman, ibra}. In this work we use Lie symmetry method to study a generalization of the Fisher equation in cylindrical coordinate.

One major model of reaction diffusion nonlinear phenomena is the Fisher equation, which in its simplest form is given by $u_t-u_{x x}=u(1-u),$ here $u(x,t)$ is a concentration of an underlying object depending upon the nature of the model \cite{Ablowitz, Britton}. There are many generalizations along with modifications of Fisher equation according to its various applications. For example, $u_t = u_{x x} -u(1-u^2)$ and $u_t = u_{x x} +u^2(1-u)$ describing neural models \cite{fitz, bokh1, huxley, Nagumo} and $u_t = u_{x x} +a u^3 +b u^2+c u$, $a,b,c\in\mathbb{R}$ describing population genetics etc., depending on values of parameters \cite{bokh1, murr}.

In this work we study one of the generalizations of Fisher equation given by $u_t=f(u)+(g(u)u_x)_x,$ which is in cylindrical coordinates
\begin{equation}\label{fisher}
u_t=f(u)+\frac{1}{x}(x g(u) u_x)_x,
\end{equation}
where $f(u),$ $g(u)$ are sufficiently smooth functions. The analysis of the generalized Fisher equation using Lie symmetry method is studied in \cite{Ovsiannikov, Cherniha1, Cherniha2, Cherniha3, Moitsheki}. The Fisher equation in cylindrical coordinates (\ref{fisher}) is studied for some combinations of $f(u)$ and $g(u)$ when $f(u)$ and $g(u)$ equals to $u^m$ or $e^{n u}$ \cite{bokh1, Bokh2, Rosa1}.

For an arbitrary $f(u)$ and $g(u)$ the equation (\ref{fisher}) has time translation symmetry \cite{Rosa1}. In section 2, We show that there are only three types of function $f(u)$ for given $g(u)=k_1 e^{k_2 u}$ so that the equation (\ref{fisher}) has Lie symmetry other than the time translation symmetry. As expected these three types of function $f(u)$ include the earlier studied case of function $f(u)$ in \cite{Rosa1}. Moreover we showed that there is an additional Lie symmetry of the equation (\ref{fisher}). Lastly, in section 3 through the symmetry reduction, we transform the generalized Fisher equation into ordinary differential equations with new independent variables.

\section{Lie symmetries of the generalized Fisher equation}
Here we study the equation (\ref{fisher}) with $f(u)\neq 0,$ $g'(u)\neq 0.$ The infinitesimal vector field is given
$X=\xi(t,x,u)\frac{\partial}{\partial x}+\tau(t,x,u)\frac{\partial}{\partial t}+\eta(t,x,u)\frac{\partial}{\partial u}.$ The second prolongation of $X$ is $
X^{(2)}=X+\mu^x\frac{\partial}{\partial u_x}+\mu^t\frac{\partial}{\partial u_t}+\mu^{x x}\frac{\partial}{\partial u_{x x}}+\mu_{t t}\frac{\partial}{\partial u_{t t}},$ where
\begin{equation}\label{2}
\mu^x=D_x(\mu)-u_x D_x(\xi)-u_t D_x(\tau), \mu^t=D_t(\mu)-u_x D_t(\xi)-u_t D_t(\tau),\mu^{x x}=D_{x}(\mu^x)-u_{x x} D_x(\xi)-u_{x t} D_x(\tau),
\end{equation}
here $D_x,$ $D_t$ are total derivatives defined as
$$
D_x=\partial_x+u_x\partial_u+u_{t x}\partial_{u_t}+u_{x x}\partial_{u_x}+\cdots,\quad D_t=\partial_t+u_t\partial_u+u_{t t}\partial_{u_t}+u_{x t}\partial_{u_x}+\cdots.
$$
If $X$ is the Lie symmetry of (\ref{fisher}) then
\begin{equation}\label{5}
X^{(2)}(u_t-f(u)-\frac{1}{x}(xg(u)u_x)_x)|_{u_t=f(u)-\frac{1}{x}(xg(u)u_x)_x}=0.
\end{equation}
Substituting (\ref{2}) into (\ref{5}) and equating the coefficients of various partial derivatives of $u$, explicitly the coefficients of $u_{x t},$ $u_x u_{x t},$ $u_x u_{x x},$ $u_{x x},$ $u_x^2,$ $u_x$ and $1$ to $0$, we get the following determining system of equations:
\begin{eqnarray}
& \tau_x = \tau_u=\xi_u=0,\label{6}\\
& 2g\xi_x-g_u\eta-g\tau_t =  0,\label{9}\\
& x^2(2g_u\xi_x-g\eta_{u u}-g_u\eta_u-g_{u u}\eta-g_u\tau_t)=0,\label{101}\\
& -x g_u\eta-x g\tau_t -2 x^2 g_u \eta_x+x g \xi_x+g \xi-x^2\xi_t-2x^2 g \eta_{u x}+x^2 g \xi_{x x}=0,\label{10}\\
& x\eta_t-x f_u\eta+x f\eta_u-x f\tau_t-x g \eta_{x x}-g \eta_x=0.\label{11}
\end{eqnarray}
From (\ref{6}) it follows that $\tau=\tau(t),$ $\xi=\xi(x,t).$  Differentiating (\ref{9}) with respect to $u$ we get
\begin{equation}\label{0916od}
2g_u\xi_x-g_{u u}\eta-g_u \eta_u-g_u\tau_t=0.
\end{equation}
Multiplying (\ref{0916od}) by $x^2$ and then subtracting it from (\ref{101}) we obtain $-x^2g\eta_{uu}=0.$ Since $g\neq 0$, we get $\eta_{uu}=0$ and $\mu=0$. Due to assumption of $g_u\neq 0$ from (\ref{9}) we get
\begin{equation}\label{12}
\eta=(2\xi_x-\tau_t)\frac{g}{g_u}, 
 \end{equation}
 and $\eta_{uu}=(2\xi_x-\tau_t)\left(\frac {g}{g_u}\right)_{uu}=0.$
Suppose $2\xi_x-\tau_t=0.$ Then $\eta=0$ and from (\ref{11}) it follows $\tau_t=0$. So $\xi_x=0$, and substitute $\eta=0$, $\xi_x=0$ and $\tau_t=0$ into (\ref{10}), we obtain $g\xi=0$. Since $g\neq 0$, we have $\xi=0$. This shows $X=\frac{\partial}{\partial t}$ \cite{Rosa1}.

Hence if the equation (\ref{fisher}) has symmetry other than $X=\frac{\partial}{\partial t}$ then $\left(\frac{g}{g_u}\right)_{u u}=0.$ From this it follows $g/g_u=c_1 u+c_2,$ $c_1, c_2\in \mathbb{R}.$
In the case of $c_1 \neq 0$ the solution is $g(u)=k_1(u+k_2)^{k_3}$ with some constants $k_1, k_2, k_3\in\mathbb{R}$ and if $c_1=0$ then the solution is $g(u)=k_1 e^{k_2 u}.$

From above it is possible that the equation (\ref{fisher}) has a symmetry other than $X=\frac{\partial}{\partial t}$ when only $g(u)=k_1(u+k_2)^{k_3}$ or $g(u)=k_1 e^{k_2 u}.$ The function $g(u)=k_1 (u+k_2)^{k_3}$ had been studied in \cite{bokh1, Bokh2,Rosa1}. So we will prove the following theorem corresponding to the partially studied function $g(u)=k_1 e^{k_2 u}.$
\begin{theorem}\label{thm1}
Let $g(u)=k_1 e^{k_2 u}$ with $k_1\neq 0, k_2\neq 0.$ Then  the equation (\ref{fisher}) has Lie symmetry other than $X=\frac{\partial}{\partial t}$ if and only if the function $f(u)$ has the following three types:
\begin{itemize}
\item[a)] If $f(u)=k_3 e^{k_4 u},$ $k_3, k_4\neq 0$ then the equation (\ref{fisher}) has symmetries
$$X_1 = \frac{\partial}{\partial t},\quad 
X_2 = (k_4-k_2)x\frac{\partial}{\partial x}+2 k_4 t\frac{\partial}{\partial t}-2\frac{\partial}{\partial u}.
$$
\item[b)] If $f(u)=k_3 e^{k_2 u}+k_5,$ $k_3, k_5\neq 0$ then the equation (\ref{fisher}) has symmetries
$$
X_1=\frac{\partial}{\partial t},\quad X_2 = e^{-k_2 k_5 t}\frac{\partial}{\partial t}+k_5 e^{-k_2 k_5 t}\frac{\partial}{\partial u}.
$$
\item[c)] If $f(u)=k_5,$ $k_5\neq 0$ then the equation (\ref{fisher}) has symmetries
$$
X_1 = \frac{\partial}{\partial t},\quad
X_2 = e^{-k_2 k_5 t}\frac{\partial}{\partial t}+k_5 e^{-k_2 k_5 t}\frac{\partial}{\partial u},\quad
X_3 = k_2 x\frac{\partial}{\partial x}+2\frac{\partial}{\partial u}.
$$
\end{itemize}
\end{theorem}
\begin{remark*}
The case a) of the theorem $g(u)=k_1e^{k_2 u},$ $f(u)=k_3e^{k_4 u}$ is the case 8 of Table 1 of \cite{Rosa1} when $k_2=k_4.$ But $k_2\neq k_4$ is not included in \cite{Rosa1}. The symmetry group of the equation (\ref{fisher}) for the case a) of the theorem corresponding to $X_2 = (k_4-k_2)x\frac{\partial}{\partial x}+2 k_4 t\frac{\partial}{\partial t}-2\frac{\partial}{\partial u}$ is $G_2 : (x,t,u)\to (e^{(k_4-k_2)\epsilon}x, e^{2k_4\epsilon}t,u-2\epsilon).$ It means that if $u=f(x,t)$ is a solution of (\ref{fisher}), then so is the function $u^{(2)} = f(e^{(k_2-k_4)\epsilon}x,e^{-2k_4\epsilon}t)-2\epsilon.$
\end{remark*}
\begin{proof}
In order to prove we solve the system of equations (\ref{10}) and (\ref{11}) considering (\ref{12}) and $g(u)=k_1 e^{k_2 u}$ for unknown functions $\tau(t),$ $\xi(x,t),$ $\eta(x,t,u)$ and $f(u).$ Substituting $g(u)=k_1 e^{k_2 u}$ and (\ref{12}) into (\ref{10}) we get $-x e^{k_2 u}\xi_x-3x^2 e^{k_2 u}\xi_{x x}+e^{k_2 u}\xi-\frac{x^2\xi_t}{k_1}=0.$
From here it follows $\xi_t=0$ and $-x\xi_{x}-3x^2 \xi_{x x}+\xi=0.$
So $\xi=\xi(x)$ and the solution is $\xi(x)=c_1 x+c_2 x^{-\frac{1}{3}}.$ Substituting (\ref{12}) and obtained $\xi(x)$ into (\ref{11}) we get
\begin{equation}
-x\tau_{t t}+\frac{2}{3}f_u c_2 x^{-\frac{1}{3}}-2 x f_u c_1+
x f_u \tau_t-x k_2 f \tau_t+\frac{32}{27}k_1 e^{k_2 u} c_2 x^{-\frac{7}{3}}=0.\label{16}
\end{equation}
The coefficient of $x^{-\frac{7}{3}}$ is zero in (\ref{16}), consequently $c_2=0.$ Eq.(\ref{16}) becomes
\begin{equation}\label{17}
\tau_{t t}+(k_2 f-f_u)\tau_t=-2 f_u c_1.
\end{equation}
If we take a derivative by $u$ and $t$ of (\ref{17}), we get $(k_2 f_u-f_{u u})\tau_{t t}=0.$ Hence there follow two subcases:$k_2 f_u -f_{u u}=0 \mbox{ or } k_2 f_u-f_{u u}\neq 0.$\\
Case 1. Let $k_2 f_u-f_{u u}=0.$ Thus $f(u)=k_3 e^{k_2 u}+k_5$ and (\ref{17}) becomes
\begin{equation}\label{odn1}
\tau_{t t}+k_2 k_5 \tau_t=-2 k_2 k_3 c_1 e^{k_2 u}.
\end{equation}
We here divide into two cases by $k_3$.
\begin{enumerate}
\item[1.1)] Let $k_3=0.$ Hence $k_5\neq 0$ and (\ref{odn1}) becomes
\begin{equation}\label{odn2}
\tau_{t t}+k_2 k_5\tau_t=0.
\end{equation}
Since $k_2 k_5\neq 0,$ the solution of (\ref{odn2}) is
$\tau=c_3 e^{-k_2 k_5 t}+c_4$ and by (\ref{12})
$\eta=(2c_1+c_3 k_2 k_5 e^{-k_2 k_5 t})/k_2,$
which corresponds to the case c) of Theorem 1.
\item[1.2)] Let $k_3\neq 0.$ Then $c_1=0$ and (\ref{odn1}) becomes
\begin{equation}\label{odn3}
\tau_{t t}+k_2 k_5\tau_t=0.
\end{equation}
The equation (\ref{odn3}) has only two possible solutions:
\begin{enumerate}
\item[1.2.1)] Let $k_5= 0.$ Then the solution (\ref{odn3}) is $\tau=c_3 t+c_4$ and again by (\ref{12}) $\eta=-\frac{c_3}{k_2},$ which corresponds to the case a) of Theorem 1, when $k_2=k_4.$
\item[1.2.2)] Let $k_5\neq 0.$ Then the solution of (\ref{odn3}) is $\tau=c_3 e^{-k_2 k_5 t}+c_4$ and $\eta=c_3 k_4e^{-k_2 k_4 t},$ which corresponds to the case b) of Theorem 1.
\end{enumerate}
\end{enumerate}
Case 2. Let $k_2 f_u-f_{u u}\neq 0.$ Then $\tau_{t t}=0.$ So $\tau(t)=c_3 t+c_4$ and (\ref{17}) becomes $(k_2 f-f_u)c_3=-2c_1 f_u.$
If $c_3=0$ then $c_1=0.$ Thus $\xi=\eta=0$ and $\tau=c_4.$ So we take $c_3\neq 0$ and obtain $f(u)=k_3 e^{k_4 u}$ with $k_3\neq 0,$ $k_2\neq k_4.$ By (\ref{17}) we have $(k_2/k_4 -1) c_3 = -2 c_1.$ 
This case corresponds to the case a) of the Theorem 1. The proof of the theorem finished here.
\end{proof}
We have successfully derived three types of $f(u)$ functions so that the equation (\ref{fisher}) has symmetry other than $X=\frac{\partial}{\partial t}$ and found Lie symmetries corresponding to each of three types of function $f(u).$

Here we would like to mention that Theorem \ref{thm1} has covered all the possibilities of function $f(u)$ of equation (\ref{fisher}) when $g(u)$ equals to $k_1 e^{k_2 u},$ compared to previously studied cases, in which were demonstrated only lists of particular cases of function $f(u)$ and $g(u)$ of equation (\ref{fisher}), so that equation (\ref{fisher}) has Lie symmetry other than $X=\frac{\partial}{\partial t}.$ Next, we find the invariant solutions of equation (\ref{fisher}). Indeed, there are cases that using Lie symmetry method is not sufficient to find the exact solution of differential equation. Even so, we can derive invariant solutions of equation (\ref{fisher}) under Lie symmetry transformation.

\section{Similarity reductions of generalized Fisher equation}
To do the reduction we first need to define an invariant solution of equation (\ref{fisher}) corresponding to symmetries found in section 2.
\begin{definition}
$u=u(x,t)$ is said to be an invariant solution of equation (\ref{fisher}) corresponding to the infinitesimal symmetry $X=\xi\frac{\partial}{\partial x}+\tau\frac{\partial}{\partial t}+\eta\frac{\partial}{\partial u}$ if and only if
\begin{itemize}
\item[i)] $u=u(x,t)$ satisfies the equation (\ref{fisher}).
\item[ii)] $u=u(x,t)$ is an invariant surface of $X,$ i.e. $\xi\frac{\partial u}{\partial x}+\tau\frac{\partial u}{\partial t}-\eta=0.$
\end{itemize}
\end{definition}
Now let us find similarity reductions and reduced equations corresponding to Theorem \ref{thm1}.
\subsection*{Case a) of Theorem \ref{thm1}}
In the case a) the equation (\ref{fisher}) is rewritten as 
\begin{equation}\label{redod}
u_t=k_3 e^{k_4 u}+\frac{1}{x}(xk_1 e^{k_2 u}   u_x)_x,\qquad k_1, k_2, k_3, k_4\neq 0.
\end{equation}
We choose the symmetry $X=(k_4-k_2)x\frac{\partial}{\partial x}+2k_4 t\frac{\partial}{\partial t}-2\frac{\partial}{\partial u}$ and will consider invariant solutions under $X.$ The characteristic equation corresponding to the symmetry $X$ is
$
\frac{d x}{(k_4-k_2)x}=\frac{d t}{2k_4 t}=\frac{d u}{-2}.
$
From here we found the similarity variable $z(x,t)=x^{2} t^{\frac{k_2}{k_4}-1}$ and the similarity transformation $u=-\frac{1}{k_4}\ln(t)+h_1(z),$ here $h_1(z)$ is a holomorphic function.  Substituting these similarity variable and transformation to the equation (\ref{fisher}) we get the desired ordinary differential equation.
\begin{theorem}
Invariant solution of equation (\ref{redod}) corresponding to the $X$ is
\begin{equation}\label{251}
u=\ln\left(\frac{h(x^{2} t^{\frac{k_2}{k_4}-1})^{\frac{1}{k_2}}}{t^{\frac{1}{k_4}}}\right),
\end{equation}
where $h(z)$ satisfies the following similarity reduction equation:
\begin{equation}\label{252}
4k_1 k_4 zh''(z)+\left(4k_1 k_4+(k_4-k_2)\frac{z}{h(z)}\right)h'(z)+k_2k_3k_4(h(z))^{\frac{k_4}{k_2}}+k_2=0.
\end{equation}
\end{theorem}
The equation (\ref{252}) is the reduced equation that we wanted to obtain. Once if we solve the reduced equation and substituting it back to (\ref{251}) we will obtain an exact solution of (\ref{redod}).\\
\textit{Example 1.} If we take $k_2=k_4=k$ in (\ref{252}) then we get the simpler form of reduced $\mbox{ODE}_{13}$ of \cite{Rosa1}
$$
4k_1zh''(z)+4k_1h'(z)+k k_3h(z)+1=0
$$
and its solution
$h(z)=c_1J\left(\sqrt{\frac{k k_3}{k_1}z}\right)+c_2Y\left(\sqrt{\frac{k k_3}{k_1}z}\right)-\frac{1}{k k_3},$ with $J$ and $Y$ Bessel functions of order 0 of first and second kind respectively.\\
\textit{Example 2.} If we take $k_1=k_3=1,$ $k_4=2k_2,$ then the reduced ordinary differential equation will be
\begin{equation}\label{odex2}
8z h''(z)+\left(8+\frac{z}{h(z)}\right)h'(z)+2k_2 h^2(z)+1=0.
\end{equation}
Even though we can not solve the equation (\ref{odex2}) exactly, now we are able to draw the solution surface by using Maple package. As an example, here we demonstrated the solution surface $u(x,t)$ of (\ref{251}) corresponding to Example 2, for two different values of $k_2.$
\begin{figure}[h]\label{fig1}
\includegraphics[width=0.5\textwidth]{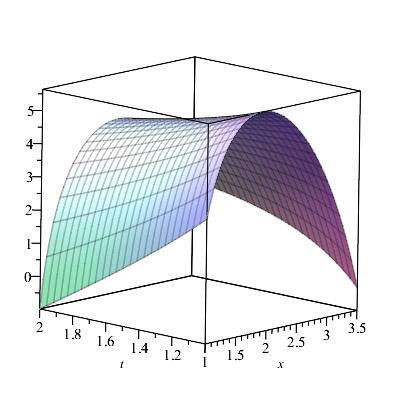}
\includegraphics[width=0.5\textwidth]{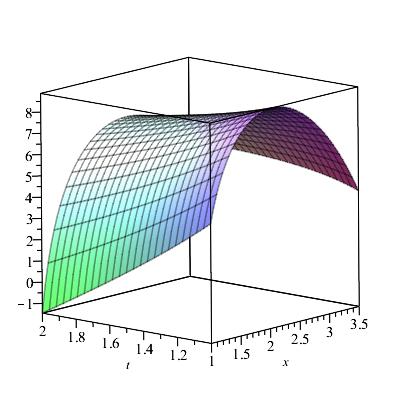}
\caption{The solution surface $u(x,t)$ given by (\ref{251}) corresponding to Example 2 with an initial condition $h(1)=2,$ $h'(1)=2.5$; left with parameter $k_2=1/4;$ right $k_2=1/6.$}
\end{figure}
\subsection*{Case b) of Theorem \ref{thm1}}
In the case b) the equation (\ref{fisher}) is rewritten as 
\begin{equation}\label{redod1}
u_t=k_3 e^{k_2 u}+k_5+\frac{1}{x}(k_1x e^{k_2 u}  u_x)_x,\qquad k_1, k_2, k_3, k_5\neq 0.
\end{equation}
Here we choose the symmetry $X=e^{-k_2 k_5 t}\frac{\partial}{\partial t}+k_5 e^{-k_2 k_5 t}\frac{\partial}{\partial u}.$
Then the characteristic equation corresponding to $X$ is
$$
\frac{d x}{0}=\frac{d t}{e^{-k_2 k_5 t}}=\frac{d u}{k_5 e^{-k_2 k_5 t}}.
$$
From here the similarity variable $z(x,t)=x$ and the transformation $u=k_5 t+h(z),$ here $h(z)$ is a holomorphic function.  Using these similarity variable and transformations to the equation (\ref{redod1}) we get the following theorem.
\begin{theorem}\label{thm3}
The invariant solution of (\ref{redod1}) corresponding to the $X$ is $u=k_5 t+\frac{1}{k_2}\ln(p(x))$, where $p(x)$ satisfies the similarity reduction equation $k_1xp''(x)+k_1p'(x)+k_2k_3xp(x)=0.$ The solution is $p(x)=c_1J\left(\sqrt{\frac{k_2k_3}{k_1}}x\right)+c_2Y\left(\sqrt{\frac{k_2k_3}{k_1}}x\right),$ with $J$ and $Y$ Bessel functions of order 0 of first and second kind respectively.
\end{theorem}
\subsection*{Case c) of Theorem \ref{thm1}}
In the case c) the equation (\ref{fisher}) is rewritten as 
\begin{equation}\label{od9}
u_t=k_5+\frac{1}{x}(k_1x e^{k_2 u} u_x)_x,\quad k_1, k_2, k_5\neq 0.
\end{equation}
According to Theorem \ref{thm3} the invariant solution corresponding to $X=e^{-k_2 k_5 t}\frac{\partial}{\partial t}+k_5 e^{-k_2 k_5 t}\frac{\partial}{\partial u}$ is $u=k_5 t+\ln(c_1+c_2\ln(x)).$ For the symmetry
$X=k_2 x\frac{\partial}{\partial x}+2\frac{\partial}{\partial u}$ the similarity transformation $u(x,t)=\ln(x^2h(t))/k_2$ can be found analogously to previous cases. So we can formulate the following theorem.
\begin{theorem}
The invariant solution of equation (\ref{od9}) corresponding to $X=k_2 x\frac{\partial}{\partial x}+2\frac{\partial}{\partial u}$ is $u(x,t)=\frac{1}{k_2}\ln(\frac{x^2}{p(t)}),$ where $p(t)$ satisfies the similarity reduction equation $p'(t)+k_2k_5p(t)+4k_1=0,$ the solution is $p(t)=-\frac{4k_1}{k_2 k_5}+c_1e^{-k_2k_5t}.$
\end{theorem}
\section{Conclusion}
In this work, we studied a generalized Fisher equation $u_t=f(u)+\frac{1}{x}(x g(u) u_x)_x$ in cylindrical coordinate by using Lie symmetry method. While doing the computation we found that there are specific types of function $f(u),$ $g(u)$ so that the generalized Fisher equation has Lie symmetry other than time translation symmetry. These obtained function types include earlier studied types of $f(u)$ and $g(u)$ of the generalized Fisher equation. In the view of the specific types of function $f(u)$ and $g(u)$ of generalized Fisher equation, we have enlarged the function types that had been studied before. Furthermore, we found the Lie symmetries corresponding to each type of function $f(u)$ and $g(u),$ then we followed the standard Lie symmetry analysis of partial differential equation.



\end{document}